\documentclass[11pt]{article}

\usepackage{booktabs}
\usepackage{array}
\usepackage{latexsym}
\usepackage{amsmath, amssymb} 
\usepackage{xspace}
\usepackage{color}
\usepackage{epic,eepic}
\usepackage{graphicx}

\usepackage{xkeyval}
\usepackage{todonotes}

\newcommand{\ra}[1]{\renewcommand{\arraystretch}{#1}}

\newtheorem{xdefinition}{Definition}
\newtheorem{xobservation}{Observation}
\newtheorem{xtheorem}{Theorem}
\newtheorem{xmaintheorem}{Main Theorem}
\newtheorem{xlemma}{Lemma}
\newtheorem{xproposition}{Proposition}
\newtheorem{xcorollary}{Corollary}
\newtheorem{xclaim}{Claim}
\newtheorem{xproperty}{Property}
\newenvironment{definition}{\begin{xdefinition}\rm}%
{\hspace*{\fill}\raisebox{-1pt}{\boldmath$\Box$}\end{xdefinition}}
{\hspace*{\fill}\raisebox{-1pt}{\boldmath$\Box$}\end{xobservation}}
\newenvironment{theorem}{\begin{xtheorem}\rm}{\end{xtheorem}}

\newenvironment{lemma}{\begin{xlemma}\rm}{\end{xlemma}}
\newenvironment{proposition}{\begin{xproposition}\rm}{\end{xproposition}}

\newenvironment{claim}{\begin{xclaim}\rm}{\end{xclaim}}

\newenvironment{proof}{\begin{trivlist}\item[]{\bf Proof }}%
{\hspace*{\fill}\raisebox{-1pt}{\boldmath$\Box$}\end{trivlist}}

\newcommand{\OPT}{\ensuremath{\operatorname{OPT}}\xspace}

\newcommand{\CEIL}[1]{\left\lceil#1\right\rceil}
\newcommand{\FLOOR}[1]{\left\lfloor#1\right\rfloor}

\newcommand{\SET}[1]{\{#1\}}
\newcommand{\SETOF}[2]{\SET{#1\;|\;#2}}

\newcommand{\SEQ}[1]{\langle #1 \rangle}
\newcommand{\A}{\ensuremath{\mathcal{A}}\xspace}
\newcommand{\B}{\ensuremath{\mathcal{B}}\xspace}
\newcommand{\I}{\ensuremath{\mathcal{I}}\xspace}
\newcommand{\LRU}{\ensuremath{\mathrm{LRU}}\xspace}
\newcommand{\FIFO}{\ensuremath{\mathrm{FIFO}}\xspace}
\newcommand{\FAR}{\ensuremath{\mathrm{FAR}}\xspace}
\newcommand{\FWF}{\ensuremath{\mathrm{FWF}}\xspace}
\newcommand{\LANG}[1]{\mathcal{L}(#1)}
\newcommand{\Ma}[2]{\ensuremath{\mathrm{Max}_{#1, #2}}}
\newcommand{\Mi}[2]{\ensuremath{\mathrm{Min}_{#1, #2}}}
\newcommand{\Min}{\ensuremath{\mathrm{Min}}}
\newcommand{\Max}{\ensuremath{\mathrm{Max}}}

\begin{document}


\title{Relative Interval Analysis of Paging Algorithms \\
on Access Graphs,\thanks{A preliminary version of this paper will appear in the
proceedings of the Thirteenth Algorithms and Data Structures Symposium. Supported in part by the Danish Council for
Independent Research. Part of this work was carried out while the
first and third authors
were visiting the University of Waterloo.}}

\author{Joan Boyar \hspace{2em} Sushmita Gupta  \hspace{2em} Kim S. Larsen \\[1ex]
        University of Southern Denmark \\
        Odense, Denmark \\[1ex]
        {\tt \{joan,sgupta,kslarsen\}@imada.sdu.dk}}

\date{\Large Draft of \today{} at \Printtime}
\date{}
\maketitle

\begin{abstract}
Access graphs, which have been used previously in connection with
competitive
analysis and relative worst order analysis to model locality of reference in
paging, are considered in
connection with relative interval analysis. The algorithms \LRU, \FIFO,
\FWF, and \FAR are compared using the path, star, and cycle access graphs.
In this model, some of the expected results are obtained.
However, although
\LRU is found to be strictly better than \FIFO on paths, it has
worse performance on stars, cycles, and complete graphs, in this model.
We solve an open question from [Dorrigiv, L{\'{o}}pez-Ortiz, Munro, 2009],
obtaining tight bounds 
on the relationship between \LRU and \FIFO with relative interval analysis.
\end{abstract}

\section{Introduction}
The paging problem is the problem of maintaining a subset of a
potentially very large set of pages from memory in a significantly
smaller cache. When a page is requested, it may already be in cache
(called a ``hit''), or it must be brought into
cache (called a ``fault'').
The algorithmic problem is the one of choosing an
eviction strategy, i.e., which page to evict from cache in the case
of a fault, with the objective of minimizing the total number of
faults.

Many different paging algorithms have been considered in the
literature, many of which can be found in \cite{BE97,DLO05j}.
Among the best known are \LRU (least-recently-used),
which always evicts the least
recently used page, and \FIFO (first-in-first-out),
which evicts pages in the order
they entered the cache.
We also consider a known bad algorithm, \FWF (flush-when-full),
which is often used for reference, since quality measures ought
to be able to determine at the very least that it is worse than
the other algorithms. If \FWF encounters a fault with a full
cache, it empties its cache, and brings the new page in.
Finally, we consider a more involved algorithm, \FAR, which
works with respect to a known access graph. Whenever a page is
requested, it is marked. When it is necessary to evict a page,
it always evicts an unmarked page. If all pages are marked in such
a situation, \FAR first unmarks all pages. The unmarked page it chooses
to evict is the one farthest from any marked page in the access graph.
For breaking possible ties, we assume the \LRU strategy in this paper.

Understanding differences in paging algorithms' behavior under
various circumstances has been a topic for much research.
The most standard measure of quality of an online algorithm,
competitive analysis~\cite{ST85j,KMRS88j}, cannot
directly distinguish between most of them.
It deems \LRU, \FIFO, and \FWF
equivalent, with a competitive ratio of $k$,
where $k$ denotes the size of the cache.
Other measures, such as relative worst order analysis~\cite{BF07,BFL07j},
can be used to obtain more separations, including
that \LRU and \FIFO are better than \FWF
and that look-ahead helps.
No techniques have been able to separate \LRU and \FIFO,
without adding some modelling of locality of reference.

Although \LRU performs better
than \FIFO in some practical situations~\cite{Y94},
if one considers all sequences
of length $n$ for any $n$, bijective/average analysis shows that
their average number of faults on these sequences is identical~\cite{ADL07},
which basically follows from \LRU and \FIFO being demand paging algorithms.
Thus, it is not surprising that some assumptions involving locality
of reference are necessary to separate them.

A separation between \FIFO and \LRU was established quite early using
access graphs for modelling locality of reference~\cite{CN99}, showing that
under competitive analysis, no matter which access graph
one restricts to, \LRU always does at least as well as \FIFO.
This proved a conjecture in~\cite{BIRS91}, where
the access graph model was introduced.
Another way to restrict the input sequences was
investigated in~\cite{AFG05}. Using Denning's working set
model~\cite{D68,D80} as an inspiration, sequences were limited with
regards to the number of distinct pages in a sliding window of size~$k$.
This also favors \LRU, as does bijective analysis~\cite{ADL07},
using the same locality of reference definition as~\cite{AFG05}.
There has also been work in the direction of probabilistic models,
including the diffuse adversary model~\cite{KP00}
and Markov chain based models~\cite{KPR00}.

The earlier successes and the generality of access graphs,
together with the possibilities the model offers with regards to
investigating specific access patterns, makes it an interesting
object for further studies.
In the light of the recent focus
on development of new performance measures,
together with the comparative studies initiated in~\cite{BIL09p},
exploring access graphs results in the context of new performance
measures seems like a promising direction for expanding our
understanding of performance measures as well as concrete algorithms.

One step in that direction was carried out in~\cite{BGL12}, where
more nuanced results were demonstrated,
showing that restricting input sequences using the access graph
model, while applying relative worst order analysis,
\LRU is strictly better than \FIFO on paths and 
cycles.
The question as to whether or not \LRU is at least as good
as \FIFO on all finite graphs was left as an open problem,
but it was shown that there exists a family
of graphs which grows with the length of the corresponding
request sequence, where \LRU and \FIFO are incomparable.
Since \LRU is optimal on paths, it is not surprising that
both competitive analysis and relative worst order analysis
find that \LRU is better than \FIFO on paths. Any ``reasonable''
analysis technique should give this result. Under competitive
analysis, \LRU and \FIFO are equivalent on cycles.
The separation by relative worst order analysis
occurs because cycles contain paths, \LRU is better
on paths, and relative worst order analysis can reflect this.
The fact that there exists an infinite family of graphs which
grows with the length of the sequence where \LRU and \FIFO are
incomparable may or may not be interesting. There are many
sequences were \FIFO is better than \LRU; they just seem to occur
less often in real applications.

Comparing two algorithms under almost any analysis technique
is generally equivalent to considering them with the
complete graph as an access graph, since the complete graph
does not restrict the request sequence in any way.
Thus, \LRU and \FIFO are equivalent on complete graphs
under both competitive analysis and relative worst order
analysis, since they are equivalent without considering
access graphs.

In this paper,  we consider relative interval
analysis~\cite{DLM09}. In some ways relative interval analysis is
between competitive analysis and relative worst order
analysis. As with relative worst order analysis, two algorithms
are compared directly to each other, rather than compared
to \OPT. This gives the advantage that, when one algorithm
dominates another in the sense that it is at least as
good as the other on every request sequence and better on
some, the analysis will reflect this. However, it is similar
to competitive analysis in that the two algorithms are always
compared on exactly the same sequence. To compare two algorithms,
\LRU and \FIFO for example, one considers the difference between
\LRU's and \FIFO's performance on any sequence, divided by the
length of that sequence. The range that these ratios can take
is the ``interval'' for that pair of algorithms. For \FIFO and
\LRU, \cite{DLM09} found two families of sequences $I_n$ and $J_n$
such that $\lim_{n\rightarrow\infty}\frac{\FIFO(I_n)-\LRU(I_n)}
{n} = -1+\frac{1}{k}$ and $\lim_{n\rightarrow\infty}\frac{\FIFO(J_n)
-\LRU(J_n)}{n}=\frac{1}{2}-\frac{1}{4k-2}$. They left it as an open problem
to determine if worse sequences exist, making the interval even
larger. In their notation, they proved:
$[-1+\frac{1}{k},\frac{1}{2}-\frac{1}{4k-2}] \subseteq \I(\FIFO, \LRU)$.
We start by proving that this is tight:
$\I(\FIFO,\LRU) = [-1+\frac{1}{k},\frac{1}{2}-\frac{1}{4k-2}]$.
These results would be interpreted as saying that \FIFO has
{\em better performance} than \LRU,
since the absolute value of 
the minimum value in the interval is larger than
the maximum, but also that they have different strengths,
since zero is contained in the interior of the interval.
We obtain more nuanced results by considering
various types of access graphs, such as paths ($P_N$),
stars ($S_N$), and cycles ($C_N$), splitting the interval of
$[-1+\frac{1}{k}, \frac{1}{2}-\frac{1}{4k-2}]$ into subintervals
for the respective graph classes.
Considering complete graphs (or cliques) implies that there are no
restrictions on the input sequences, so this is equivalent to
considering the situation without an access graph.
Table~\ref{comparisons} shows our results.

\begin{table}[h!]
\caption{Summary of Results}
\label{comparisons}
\medskip
\centering
\ra{1.2}
\small
\begin{tabular}{|c@{\hspace{0.35em}}c@{\hspace{0.35em}}c@{\hspace{0.35em}}c@{\hspace{0.35em}}c|c|c|}
\hline
Lower Bound & &Relative Interval && Upper Bound &
Th. \\
\hline 
& & $\I[\FIFO, \LRU]$ & $=$ & $ \left[ -1+\frac{1}{k},
~\frac{1}{2} - \frac{1}{4k-2} \right] $ & \ref{thm:G:Main-result}\\
&   & $\I[\FWF, \mathcal{A}]$ & $=$ & $\left[0, ~1-
\frac{1}{k}\right]$ & \ref{FWF:P} \\
$\left[0, ~1 - \frac{k+1}{k^2} \right]$ & $\subseteq$ &
$\I[\FWF, \FIFO]$ & $\subseteq$  & $ \left[0,~ 1-
\frac{1}{k}\right]$& \ref{FWF-FIFO}  \\
\hline
&  & $\I^{P_{N}}[\FIFO, \A]$ & $=$     & $ \left[0, ~
\frac{1}{2}-\frac{1}{2k}\right]$ & \ref{thm:path} \\
& & $\I^{P_{N}}[\FWF, \A]$ & $=$  & $ \left[0,~1 - \frac{1}{k}\right]$
& \ref{FWF:P}\\
$\left[ 0, ~1 - \frac{k+1}{k^2} \right]$ & $\subseteq$ &
$\I^{P_{N}}[\FWF, \FIFO]$ & $\subseteq$  & $ \left[0, ~1-
\frac{1}{k}\right]$&  \ref{FWF-FIFO}\\
\hline
$ \left[ -\frac{1}{2}+ \Theta(\frac{1}{k}),~ \frac{1}{4} +
\Theta(\frac{1}{k}) \right]$ & $ \subseteq $ &  $\I^{S_N}[\FIFO, \A]$
&$\subseteq $ & $\left[ -\frac{1}{2}+\Theta(\frac{1}{k}),~ \frac{1}{4}
+ \Theta(\frac{1}{k}) \right] $ & \ref{thm:S:FIFO} \\
&  & $\I^{S_{N}}[\FWF, \mathcal{B}]$ & $=$     & $ \left[0,
~\frac{1}{2}\right]$ & \ref{thm:Max:Star-FWF}  \\ 
\hline
$\left[ -1+\frac{r}{k}, ~\frac{1}{2}-\frac{1}{4k-2} \right] $ & $
\subseteq $& $ \I^{C_{N}}[ \FIFO, \LRU] $& $ \subseteq$ & $\left[
-1+\frac{1}{k}, ~\frac{1}{2}-\frac{1}{4k-2} \right]$ &
\ref{cycle-FIFO-LRU} \\
 &  & $\I^{C_{N}}[\FWF, \LRU]$ & $ = $ & $ \left[ 0, ~1
-\frac{1}{k}\right]$ & \ref{Max:C:FWF} \\
$\left[- \frac{r\left(  \FLOOR{ \log \frac{\hat{N}}{r} } -1 \right)}{N-1},   1-\frac{X_r}{k}  \right ] $ & $ \subseteq$ &
$\I^{C_{N}}[\LRU, \FAR] $ & $\subseteq$ &
$\left[- \frac{X_r -1}{k}, ~1 -\frac{1}{k} \right]$  & \ref{thm:C:FAR}
 \\
$\left[0,  ~ 1 -\frac{X_{r}}{k} \right] $ & $ \subseteq$ &
$\I^{C_N}[\FWF, \FAR]$ & $ \subseteq $ & $\left [0, ~1-
\frac{1}{k}\right]$ & \ref{thm:C:FAR}\\
$\left[0, ~1 - \frac{k+1}{k^2} \right]$ & $\subseteq$ &
$\I^{C_{N}}[\FWF, \FIFO]$ & $\subseteq$  & $ \left[0,~ 1-
\frac{1}{k}\right]$& \ref{FWF-FIFO}  \\
\hline
\end{tabular}
\begin{center}
${\A} \in \{\FAR, \LRU  \}$ and $\mathcal{B} \in \{\FAR, \FIFO, \LRU
\}$.\\
$N=k+r$, with $1 \leq r\leq k-1$, $ X_r=
r(x-1) + \CEIL{\frac{N}{2^x}}$ with $x=  \FLOOR{\log \frac{N}{r}}$.\\
$\hat{N}$ denotes $N$ if $N$ is even, and $N-1$ otherwise.
\end{center}
\end{table}

Comparing these results with the results from competitive analysis
and relative worst order analysis, both with respect to access graphs,
it becomes clear that different measures highlight different
aspects of the algorithms. All the measures show that \LRU is strictly
best on paths, which is not surprising since it is in fact optimal
on paths and \FIFO is not.
On the other access graphs considered here,
relative interval analysis gives results which can be interpreted as
incomparability, but leaning towards deeming \FIFO
the better algorithm. Relative worst order analysis, on the other
hand, shows that on cycles, \LRU is strictly better than \FIFO, and
on complete graphs, they are equivalent. It has not yet been studied
on stars, but an incomparability result for \LRU and \FIFO has been
found for a family of graphs growing with the length of the input.

\section{Preliminaries}
We have defined the paging algorithms in the introduction.
If more detail is desired, the algorithms are described in~\cite{BE97}.

An {\em access graph} for paging models the access patterns,
i.e., which pages can be requested after a given page.
Thus, the vertices are pages, and after a page $p$ has been
requested, the next request is to $p$ or one of its neighbors
in the access graph.
We let $N$ denote the number of vertices of the access graph under
consideration at a given time.
This is the same as the number of different pages we consider.
We will always assume that $N>k$, since otherwise the problem is trivial,
and let $r=N-k$.
A requests sequence is a sequence of pages and the
sequence {\em respects} a given access graph
if any two consecutive requests are either identical or
neighbors in the access graph.
We let $\LANG{G}$ denote the set of all request sequences respecting $G$.

We use the definition of $k$-phases from~\cite{BE97}:
\begin{definition}
A request sequence can be divided recursively into a number of
{\em $k$-phases} as follows:
Phase~$0$ is the empty sequence.
For every $i \geq 1$, Phase~$i$ is a maximal sequence following
Phase~ $i-1$ containing at most $k$ distinct requests.
\end{definition}
Thus, Phase~$i$ begins on the $(k+1)$st distinct page requested
since the start of Phase~$i-1$, and the last phase may contain fewer than
$k$ different pages. We generally want to ignore Phase~0, and
refer to Phase~1 as the first phase.

Similarly, we can define $x$-blocks, for some integer $x$,
focusing on when a given algorithm \A has faulted $x$ times.
\begin{definition} 
A request sequence can be divided recursively into a number of
{\em $x$-blocks} with respect to an algorithm \A as follows:
The $0$th $x$-block is the empty sequence.
For every $i \geq 1$, the $i$th $x$-block is a maximal sequence following
the $(i-1)st$ $x$-block for which \A faults at most $x$ times.

The {\em complete} blocks are defined to be the ones with
$x$ faults, i.e., excluding the $0$th block and possibly the last.
\end{definition}

There are some well-known and important classifications of
paging algorithms, which are used here and in most other
papers on paging~\cite{BE97}:
An paging algorithm is called {\em conservative}
if it incurs at most $k$ page faults on any consecutive subsequence
of the input containing $k$ or fewer distinct page references.
\LRU and \FIFO belong to this class.
Similarly, a paging algorithm is called a {\em marking} algorithm
if for any $k$-phase, once a page has been requested in that phase,
it is not evicted for the duration of that phase.
\LRU, \FAR, and \FWF are marking algorithms.

If \A is a paging algorithm, we let $\A(I)$ denote \A's cost
(number of faults) on the input (request) sequence $I$.
We now adapt relative interval analysis from~\cite{DLM09} to access graphs.
Let \A and \B be two algorithms. We define the following notation:

\[
 \Mi{\A}{\B}(n,G) = \min_{|I|=n, I \in\LANG{G}} \SET{ \A(I) - \B(I) }
\]

\[
 \Ma{\A}{\B}(n,G) = \max_{|I|=n, I \in\LANG{G}} \SET{ \A(I) - \B(I) }
\]

\[
  \Min^G(\A,\B) = \lim_{n\rightarrow\infty}\inf \frac{\Mi{\A}{\B}(n,G)}{n}
\]

\[
  \Max^G(\A,\B) = \lim_{n\rightarrow\infty}\sup \frac{\Ma{\A}{\B}(n,G)}{n}
\]

\begin{definition}
The {\it relative interval} of two algorithms \A and \B
with respect to the access graph, $G$, is
\[ \I^G(\A,\B)  = [\Min^G(\A,\B),\Max^G(\A,\B)] \] 

\B \emph{has better performance than} \A if $\Max^G(\A,\B) > |\Min^G(\A,\B)|$.

\B \emph{dominates} \A if $\I^G(\A, \B) = [0, \beta]$ for some $\beta > 0$.
%
\end{definition}
Note that in the above, $\Max^G(\A,\B) = -\Min^G(\B,\A)$.

This definition generalizes the one from~\cite{DLM09} in that
the original definition is the special case where $G$ is the complete graph,
which is the same as saying that there are no restrictions on the
sequences. We omit $G$ in the notation when $G$ is complete.

Note that if \B \emph{dominates} \A, this means that
\A does not outperform \B on any sequence (asymptotically), 
while there are sequences on which \B outperforms \A.
Also, when $\Max^G(\A, \B)$ is close to $0$,
this indicates that \A's performance is not much worse than that of \B's.

The following general lemmas will prove helpful later.
The first observation is well known for $k$-phases~\cite{BE97}:
\begin{lemma}
\label{minimum-cost}
Any algorithm has at least $b + k - 1$ faults on a sequence
consisting of $b$ complete $k$-phases or
$b$ complete $k$-blocks defined with respect to any conservative or marking
algorithm.
\end{lemma}
\begin{proof}
Let $p$ be the page requested first in Phase~$i$ and let
$I'$ be the subsequence starting with the second request in Phase~$i$
and ending right after the first request in Phase~$i+1$.
Since there are $k$ different pages in $I'$ different from $p$,
and $p$ is in cache right after it has been processed,
any algorithm must fault at least once in $I'$.
Thus, an algorithm must fault at least $k+1$ times on Phase~$1$ and
the first request in Phase~$2$, and then at least once for the
next $b-2$ $k$-phases, summing to $b + k - 1$.

The only properties used in the above are the following: First, there
are at least $k$ distinct requests in a $k$-phase, and, second,
for any phase, the first request is different from any request
in the previous phase;
specifically, the first request in two subsequent $k$-phases are different.
Any conservative or marking algorithm gives rise to such $k$-blocks.
\end{proof}

\begin{lemma}
\label{phases-limit-max}
Assume that for two algorithms \A and \B,
there exist functions $f$ and $g$ such that
\begin{itemize}
\item
$\lim_{n\rightarrow\infty} \Ma{\A}{\B}(n,G)=\infty$,
\item
for all $I\in\LANG{G}$,
$\A(I)-\B(I) \leq f(b_I)$
and
$|I|\geq g(b_I)$,
where $b_I$ denotes the number of complete $k$-phases or $k$-blocks in $I$,
and the $k$-blocks are defined with respect to a conservative or
marking algorithm,
and
\item
the limit $\lim_{b\rightarrow\infty}\frac{f(b)}{g(b)}$ exists.
\end{itemize}
Then
$\Max^G(\A, \B) \leq \lim_{b\rightarrow\infty}\frac{f(b)}{g(b)}$.
\end{lemma}
\begin{proof}
In this proof, we will take the word ``phase'' to mean either a
$k$-phase or a $k$-block.

We define a sequence of request sequences as follows.
For $j\geq 1$, let $I_j$ be a sequence of length~$j$
such that $I_j$ maximizes $\A(I) - \B(I)$ over all sequences of length~$j$.

By construction,
$\Ma{\A}{\B}(n,G)= \max_{|I|=j, I \in\LANG{G}} \SET{ \A(I) - \B(I) } = \A(I_j)-\B(I_j) \leq f(b_{I_j})$,
and by assumption, $|I_j|\geq g(b_{I_j})$. Thus
$\frac{ \Ma{\A}{\B}(n,G) }{|I_j|} \leq \frac{f(b_{I_j})}{g(b_{I_j})}$.
Now,
\[
  \Max^G(\A,\B)
\leq \limsup_{j\rightarrow\infty}\frac{f(b_{I_j})}{g(b_{I_j})}
=    \limsup_{b\rightarrow\infty}\frac{f(b)}{g(b)}
=    \lim_{b\rightarrow\infty}\frac{f(b)}{g(b)}.
\]
The second to last equality holds since
$\SETOF{b_{I_j}}{j \geq 1}$ contains infinitely many values.
Assume to the contrary that it had a maximum value $b_{I_{j'}}$ for some $j'$.
That would mean that for any $j$,
$\Ma{\A}{\B}(|I_j|,G)\leq\max\SETOF{f(b)}{1\leq b\leq b_{I_{j'}}}$,
contradicting the assumption of the left-hand
expression being unbounded.

The last equality holds since we have assumed that the limit exists.
\end{proof}

The proof of the following is analogous to the lemma just proven.
Note, however, that the function $f$ in the second bullet
has image in $\mathbb{R}^-$.
\begin{lemma}
\label{phases-limit-min}
Assume that for two algorithms \A and \B,
there exist functions $f$ and $g$ such that
\begin{itemize}
\item
$\lim_{n\rightarrow\infty}\Min_{\A,\B}(n,G)=-\infty$,
\item
for all $I\in\LANG{G}$,
$\A(I)-\B(I) \geq f(b_I)$
and
$|I|\geq g(b_I)$,
where $b_I$ denotes the number of complete $k$-phases or $k$-block in $I$, 
and the $k$-blocks are defined with respect to a conservative or
marking algorithm, and
\item
the limit $\lim_{b\rightarrow\infty}\frac{f(b)}{g(b)}$ exists.
\end{itemize}
Then
$\Min^G(\A, \B) \geq \lim_{b\rightarrow\infty}\frac{f(b)}{g(b)}$.
\end{lemma}

\section{Complete Graphs}

As remarked earlier, if the access graph is complete, it incurs no
restrictions, so the result of this section is in the same model
as~\cite{DLM09}.
In~\cite{DLM09}, it is shown that
$[-\frac{k-1}{k}, \frac{k-1}{2k-1}] \subseteq \I(\FIFO, \LRU)$.
Below, we answer an open question from~\cite{DLM09},
proving that this is tight.

\begin{lemma}
\label{any-fifo-lru}
For any access graph $G$,
\[
   -1+\frac{1}{k} \leq \Min^G(\FIFO, \LRU)
\mbox{ and }
   \Max^G( \FIFO, \LRU) \leq \frac{1}{2} - \frac{1}{4k-2}.
\]
\end{lemma}

\begin{proof}
We first consider the $\Min$ value.
Suppose that a sequence $I$ has $b$ complete $k$-phases.
Since \LRU is conservative and a complete $k$-phase contains
$k$ distinct pages, it cannot fault more than $bk+k-1$ times~\cite{BE97}.
By Lemma~\ref{minimum-cost}, $\FIFO(I) \geq k + b - 1$.
Thus, $\FIFO(I) - \LRU(I) \geq k+b-1-(bk+k-1)=-b(k-1)$.
Each $k$-phase must have length at least $k$, and
$\lim_{b\rightarrow\infty}\frac{-b(k-1)}{bk}=-\frac{k-1}{k}$.
Clearly, $\min_{|I|=n, I \in\LANG{G}} \SET{ \FIFO(I) - \LRU(I) }$
goes towards $-\infty$ as a function of $n$
(see for instance the family of sequences $J_n$ from Lemma~\ref{cycle-FIFO-LRU-min-ub}).
Thus, by Lemma~\ref{phases-limit-min},
$\Min^G(\FIFO,\LRU)\geq-\frac{k-1}{k}=-1+\frac{1}{k}$.

We now consider the $\Max$ value.
Given a request sequence $I$, we let
$B_i$ denote the $i$th $k$-block for \FIFO.
Assume that there are $b$ complete $k$-blocks.
\FIFO faults $k$ times per complete $k$-block and up to $k-1$ times
for the possible final $k$-block.
Thus, $\FIFO(I) \leq bk + (k-1)$.
Assume that \LRU faults $\alpha_i$ times in $B_i$.
By Lemma~\ref{minimum-cost},
\LRU faults at least $b + k - 1$ times.
Thus, $\Sigma_{i=1}^b\alpha_i\geq b+k-1$.

We now compute a lower bound on the length of the request sequence $I$
based on the number of complete $k$-blocks in it and the algorithms'
behavior on it.

As a first step, with every request on which \FIFO faults and \LRU has
a hit, we associate a distinct request where \FIFO has a hit.
Let $r$ be such a request to a page $p$ in $B_i$.
Since it is a hit for \LRU, $p$ must have been requested in the
maximal subsequence of requests $I'$ consisting of $k$ distinct pages
and ending just before $r$.
Consider the first such request, $r'$, in $I'$.
If it were a fault for \FIFO, \FIFO could not have faulted
again on $r$. Thus, $r'$ was a hit for \FIFO and we associate $r'$ with $r$.

To establish that the association is distinct, assume that $r'$ also gets
associated with a request $r''$. Without loss of generality, assume that
$r''$ is later than $r$. For \FIFO to fault on both $r$ and $r''$, there
must be at least $k$ distinct pages different from $p$ in between $r$
and $r''$. However, since we are assuming that \LRU has a hit on $r''$,
by the property of \LRU, the page requested by $r''$
must have been requested during the
same $k$ distinct pages. Thus, by the
construction above, the page that gets associated with $r''$ (and $r$)
will be later than $r$, which is a contradiction.

Thus, if \LRU faults $\alpha_i$ times in $B_i$, by the procedure above,
we identify at least $k-\alpha_i$ distinct requests.
In total, there are at least $\Sigma_{i=1}^b (k-\alpha_i) =
bk-\Sigma_{i=1}^b \alpha_i$ distinct hits for
\FIFO in $I$ and, since there are $b$ complete $k$-blocks, at least
$bk$ faults. Thus, the length of $I$ is at least
$2bk-\Sigma_{i=1}^b \alpha_i$,
and
\[
 \frac{\FIFO(I) - \LRU(I)}{|I|}
 \leq \frac{bk + k-1 - \Sigma_{i=1}^b\alpha_i}{2bk-\Sigma_{i=1}^b \alpha_i}.
\]
By the lower bound on $\Sigma_{i=1}^b\alpha_i$ above, and
the arithmetic observation that $\frac{u-y}{v-y}<\frac{u-x}{v-x}$,
if $u<v$ and $x<y<v$, we have that
\[
\frac{bk + k-1 - \Sigma_{i=1}^b\alpha_i}{2bk-\Sigma_{i=1}^b \alpha_i}
\leq \frac{bk + k-1 - (b+k-1)}{2bk-(b+k-1)}
= \frac{b(k-1)}{b(2k-1)-k+1}.
\]

Clearly, $\max_{|I|=n, I \in\LANG{G}} \SET{ \FIFO(I) - \LRU(I)) }$ is unbounded 
as a function of $n$
(see for instance the family of sequences $I_n$ in Lemma~\ref{Star-Max}).
By Lemma~\ref{phases-limit-max},
$\Max^G(\FIFO,\LRU)\leq\frac{k-1}{2k-1}=\frac{1}{2} - \frac{1}{4k-2}$,
since
$\lim_{b\rightarrow\infty}\frac{b(k-1)}{b(2k-1)-k+1}=\frac{k-1}{2k-1}$.
\end{proof}

From~\cite{DLM09} and Lemma~\ref{any-fifo-lru}, we have the following:
\begin{theorem}\label{thm:G:Main-result}
$\I(\FIFO, \LRU) = [-1+\frac{1}{k}, \frac{1}{2} - \frac{1}{4k-2}]$.
\end{theorem}

The following gives general bounds that are applicable to all pairs
of algorithms considered here, though in many cases better bounds are
proven later. The proof was essentially given in the
first paragraph of the proof of Lemma~\ref{any-fifo-lru}.

\begin{proposition}\label{general_bound}
Let \A be a conservative or marking algorithm and \B be any algorithm
for paging, then for any access graph $G$, $\Min^G [ \B,\A] \geq  
-1+\frac{1}{k} $ and
$\Max^G [ \A,\B] \leq 1-\frac{1}{k} $.
\end{proposition}

\subsection{FWF}

\FWF performs very badly compared to the other algorithms
considered here, \LRU, \FAR, and \FIFO. 
The following is folklore:
\begin{lemma}\label{fwf-conservative}
For any sequence $I$ and any conservative or marking algorithm \A,
we have $\A(I) \leq \FWF(I)$. 
\end{lemma}

This implies that for any access graph $G$, $\A^G(I) \leq \FWF^G(I)$ and so
\[\Min^G[\FWF, \LRU] = \Min^G[\FWF,\FIFO] = \Min^G[\FWF,\FAR] = 0.\]
Thus, \LRU, \FIFO, and \FAR all dominate \FWF.

The upper bound of $1-\frac{1}{k}$ from Proposition~\ref{general_bound} is tight for \FWF versus either \LRU, for any
access graph containing a path on $k+1$ vertices, and it is
tight for \FWF versus \FAR on a path containing at least
$k+1$ vertices. Note that a cycle on $k+1$ vertices contains
a path on $k+1$ vertices, but \FAR does not behave identically
on these two graphs.

\begin{theorem}\label{FWF:P}
For the path access graph $P_N$, where $N\geq k+1$ (and for \LRU for any graph
containing $P_{k+1}$), and $\A \in \{ \LRU,  \FAR\}$,
\[ \I^{P_N}[\FWF, \A] = \left[0, 1 - \frac{1}{k} \right]. \]
\end{theorem}

\begin{proof}
Consider the sequence $I_n = \SEQ{1,2,\ldots, k, k+1, k, \ldots, 2}^n$.
For this we have $\LRU(I_n) =\FAR(I_n) =2n +k-1$, and $\FWF(I_n)=2kn$.
Therefore, 
\[\lim_{n \rightarrow \infty}\frac{\FWF(I_n) - \LRU(I_n)}{|I_n|} = 
\lim_{n \rightarrow \infty}\frac{\FWF(I_n) - \FAR(I_n)}{|I_n|} = 
\frac{k-1}{k}.
\]

By Proposition~\ref{general_bound}, this gives $\Max^{P_N}(\FWF, \LRU) = 
\Max^{P_N}(\FWF, \FAR) =  1 - \frac{1}{k}$. Lemma~\ref{fwf-conservative} shows that 
\LRU and \FAR dominate \FWF. 
\end{proof}

The same tight result for \FWF versus \FIFO almost holds.

\begin{theorem}\label{FWF-FIFO}
For any graph $G$ containing a path with $k+1$ vertices, 
if $k$ is odd, then
\[\I^{G}[\FWF, \FIFO] = \left[0, 1-\frac{1}{k} \right],\]
and if $k$ is even, then
\[\left[0,\frac{k^2-k-1}{k^2}\right]\subseteq\I^G[\FWF, \FIFO] \subseteq \left[0, \frac{k-1}{k} \right].\]
\end{theorem}
\begin{proof}
Let $h=\lfloor (k+1)/2 \rfloor$.
Define the subsequence
\[S_i = \SEQ{ h+i,h+i-1,...,h,...,h-i,h-i+1,...,h,...,h+i}\]
and define the subsequence $R$ which starts with page $h$ and
then requests $S_1,S_2,...,S_{h-1}$.
This initial part of every sequence in our family of sequences ensures
that \FIFO's order for faulting is always $\SEQ{h,h+1,h-1,h+2,h-2,...,2h-1,1}$.
The value $2h-1$ is $k$ if $k$ is odd and $k-1$ if $k$ is even.

Suppose $k$ is odd.
Let $I_n =\langle R,K_n\rangle$, where
$J = \SEQ{ k+1,k,...,1,2,...k}^{h}$ and
$K_n = J^n$.
\FWF and \FIFO fault the same number of times on $R$.
\FWF faults $2khn$ times on  $K_n$.
On the first request to $k+1$ in $I_n$, \FIFO evicts $h$. Thus,
after the fault on $k+1$, its only fault while going ``left'' 
(towards lower page numbers) for the
first time in $J$ is on $h$, and its only fault going ``right'' is
on $h+1$. On the $i$th iteration ($i\leq h-1$) of $J$, it faults on $h-i+1$
going left and on $h+i$ going right. On iteration $h$, it only faults on
$1$, so \FIFO has the same cache configuration immediately after $J$
as it had immediately before. Thus, \FIFO has $k+1$ faults on $J$,
giving $(k+1)n$ in all. The number of requests in $K_n$ is 
$2kn$. Thus, $\lim{n\rightarrow \infty}\frac{\FWF(I_{n}) - \FIFO(I_{n})}
{|I_{n}|} = \frac{2khn-(k+1)n}{2khn}=\frac{k-1}{k}$.

Suppose $k$ is even. We define similar sequences, but
let $I_n=\langle R,k,K_n\rangle$, since $k$ is not requested yet.
\FIFO will still fault $k+1$ times on $J$, but
\[\lim{n\rightarrow \infty}\frac{\FWF(I_{n}) - \FIFO(I_{n})}
{|I_{n}|} = \frac{2khn-(k+1)n}{2khn}=\frac{k^2-k-1}{k^2}.\]
Lemma~\ref{fwf-conservative} shows that \FIFO dominates \FWF.
\end{proof}

\section{Path Graphs}
In this section, we analyze path access graphs, $P_N$, with $N$ vertices.
We assume that $N\geq k+1$, since otherwise, results become trivial.
\begin{lemma} \label{path-upper-bound}
For the path access graph $P_N$,
$$\Max^{P_N}(\FIFO,\LRU)\leq\frac{1}{2}-\frac{1}{2k}.$$
\end{lemma}
\begin{proof}
Consider any request sequence $I$.
We divide the sequence up into phases as described now
(these are {\em not} $k$-phases).
Initially, define a direction by where \LRU makes its $k$th
fault compared with its cache content. Without loss of generality,
we assume this happens going to the right on the path.

We start the first phase with the first request and later explain how
subsequent phases are started.
In all the phases, we start to the left (relatively).
In all phases, except the first, \LRU has the first
$k-1$ distinct pages that will be requested during that phase in cache.
In all phases, the first fault by \LRU in the phase, after having 
processed the first $k-1$ distinct pages, is to the right.
We maintain this as an invariant that holds at the start of any phase,
though the direction can change, as we will get back to at the
end of the proof.
The exception in the first phase, adding an extra $k-1$ faults
to the cost of \LRU as compared with the analysis below,
will not influence the result in the
the limit for the length of the request sequence going towards infinity.

We want to analyze a phase where \LRU faults to the right before
it faults to the left again. These faults to the right may not
appear consecutively. There may be some faults in a row, but then
there may be hits and then faults again, etc. Thus,
assume that there are $m$ maximal subsequences of requests to the
right where \LRU faults---all of this before \LRU faults
going to the left again.
Assume further that these maximal subsequences of requests give rise to
$s_1, s_2, \ldots, s_m$ faults, respectively, where, by definition,
$m\geq 1$, and let $s=\Sigma_{i=1}^m s_i$.

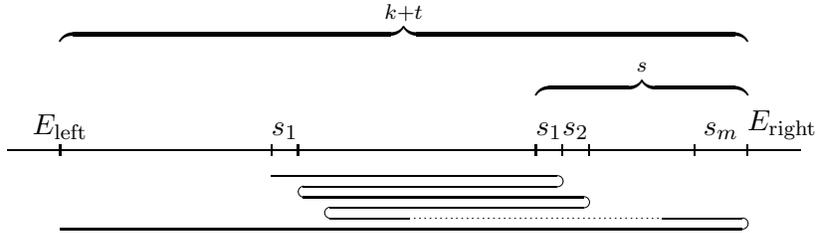
\medskip
\begin{center}
\begin{picture}(300,82)(0,-30)
\put(0,0){\line(1,0){300}}
\put(20,-2){\line(0,1){4}}
\put(20,4){\makebox(0,0)[b]{$E_{\mathrm{left}}$}}
\put(280,-2){\line(0,1){4}}
\put(280,4){\makebox(0,0)[bl]{$E_{\mathrm{right}}$}}
\put(150,40){\makebox(0,0)[b]{$\overbrace{\makebox(260,0)[b]{}}^{k+t}$}}
\put(240,20){\makebox(0,0)[b]{$\overbrace{\makebox(80,0)[b]{}}^{s}$}}
\put(100,-2){\line(0,1){4}}
\put(105,4){\makebox(0,0)[b]{$s_1$}}
\put(110,-2){\line(0,1){4}}
\put(200,-2){\line(0,1){4}}
\put(205,4){\makebox(0,0)[b]{$s_1$}}
\put(210,-2){\line(0,1){4}}
\put(215,4){\makebox(0,0)[b]{$s_2$}}
\put(220,-2){\line(0,1){4}}
\put(260,-2){\line(0,1){4}}
\put(270,4){\makebox(0,0)[b]{$s_m$}}
\put(100,-10){\line(1,0){108}}
\put(112,-14){\line(1,0){96}}
\put(112,-18){\line(1,0){106}}
\put(122,-22){\line(1,0){96}}
\put(122,-26){\line(1,0){30}}
\path(152,-26)(152,-26)
\path(154,-26)(154,-26)
\path(156,-26)(156,-26)
\path(158,-26)(158,-26)
\path(160,-26)(160,-26)
\path(162,-26)(162,-26)
\path(164,-26)(164,-26)
\path(166,-26)(166,-26)
\path(168,-26)(168,-26)
\path(170,-26)(170,-26)
\path(172,-26)(172,-26)
\path(174,-26)(174,-26)
\path(176,-26)(176,-26)
\path(178,-26)(178,-26)
\path(180,-26)(180,-26)
\path(182,-26)(182,-26)
\path(184,-26)(184,-26)
\path(186,-26)(186,-26)
\path(188,-26)(188,-26)
\path(190,-26)(190,-26)
\path(192,-26)(192,-26)
\path(194,-26)(194,-26)
\path(196,-26)(196,-26)
\path(198,-26)(198,-26)
\path(200,-26)(200,-26)
\path(202,-26)(202,-26)
\path(204,-26)(204,-26)
\path(206,-26)(206,-26)
\path(208,-26)(208,-26)
\path(210,-26)(210,-26)
\path(212,-26)(212,-26)
\path(214,-26)(214,-26)
\path(216,-26)(216,-26)
\path(218,-26)(218,-26)
\path(220,-26)(220,-26)
\path(222,-26)(222,-26)
\path(224,-26)(224,-26)
\path(226,-26)(226,-26)
\path(228,-26)(228,-26)
\path(230,-26)(230,-26)
\path(232,-26)(232,-26)
\path(234,-26)(234,-26)
\path(236,-26)(236,-26)
\path(238,-26)(238,-26)
\path(240,-26)(240,-26)
\path(242,-26)(242,-26)
\path(244,-26)(244,-26)
\path(246,-26)(246,-26)
\path(248,-26)(248,-26)
\put(248,-26){\line(1,0){30}}
\put(20,-30){\line(1,0){258}}
\put(112,-16){\oval(4,4)[l]}
\put(122,-24){\oval(4,4)[l]}
\put(208,-12){\oval(4,4)[r]}
\put(218,-20){\oval(4,4)[r]}
\put(278,-28){\oval(4,4)[r]}
\end{picture}
\end{center}
\medskip

For now, we assume that for all $i$, $s_i<k$.
Thus, \LRU moves left and right at least $m$ times;
maybe more times where it does not give rise to faults.
Since it does not fault going to the left during these turns,
the faults are to pages further and further to the right.
Let $E_{\mathrm{right}}$ denote the extreme rightmost position
it reaches during these faults to the right.

When \LRU faults again to the left after having processed $E_{\mathrm{right}}$,
we consider the leftmost node
$E_{\mathrm{left}}$, where \LRU faults after the $s$ faults described above,
but before it faults to the right again.
We end the phase with the first request to $E_{\mathrm{left}}$ after
the $s$ faults.
We define
subsequent phases inductively in the same way, starting with the first request
not included in the previous phase,
possibly leaving an incomplete phase at the end.

We now consider the costs of the algorithms and the length of the
sequence per phase.
\LRU faults $s$ times going to the right during the $m$ turns in the phase.
Additionally, \LRU must fault at least $t$ times going from
$E_{\mathrm{right}}$ to $E_{\mathrm{left}}$, where $t$ is defined by there
being $k+t$ nodes between $E_{\mathrm{left}}$ and $E_{\mathrm{right}}$,
including both endpoints.
This sums up to $s+t$ faults.

For \FIFO, we postpone the discussion of the first $s_1$ distinct
pages seen in a phase. Just to avoid any confusion, note that
these pages are immediately to the right of $E_{\mathrm{left}}$
(the endpoint of the previous phase) and thus not the pages that \LRU faults on.
After that, consider the maximal subsequence
of at most $k$ distinct pages. This subsequence starts with the
$(s_1+1)$st distinct request (the last request to it before the $s_2$
faults) and continues up to, but not including
the first request that \LRU has one of its $s_2$ faults on.
We know that there are at most $k$ pages there, because
\LRU  only faults $s_1$ times there.
Assume that \FIFO faults $f_1$ times on this subsequence.
Since \FIFO is conservative, $f_1\leq k$.

We define more such
subsequences repeatedly, the $(m-1)$st of these ending just before
\LRU's first fault of the $s_m$ faults, and the $m$th including the
$s_m$ faults and $k$ of the $k+t$ nodes before we reach $E_{\mathrm{left}}$.
Finally, we return to the question of the first $s_1$ distinct
pages seen in the phase. These overlap with the
``$t$ pages'' from the previous phase; otherwise we would not have started
the phase where we did. If \FIFO faults on one of these pages
when going through the $t$ pages in the previous phase, it will
not fault on them again in this phase. Thus, we only have to count them in
one phase, and choose to do this in the previous phase.
In total, \FIFO faults at most $(\Sigma_{i=1}^m f_i) + t$ times,
and for all $i$, $f_i\leq k$.

The difference between the cost of \FIFO and \LRU is then at most
$(\Sigma_{i=1}^m f_i) + t - (s + t)= (\Sigma_{i=1}^m f_i) - s
= (\Sigma_{i=1}^m (f_i - 1)) - (s - m)$.

From the analysis of \FIFO above, knowing that on a subsequence
of length at most $k$, \FIFO can fault at most once on any given page,
if it faults $f_i$ times, the subsequence has at least $f_i$
distinct pages.
Given that the subsequence starts at the left end of the ``$s_i$ pages''
and ends at the right end of the ``$s_i$ pages'', all pages that
\FIFO faults on, except possibly the leftmost, must be requested at
least twice, giving at least $2f_i - 1$ requests.
So, the length of the sequence is at least
$(\Sigma_{i=1}^m (2f_i - 1)) + t$.
We now sum up over all phases, equipping each variable with
a superscript denoting the phase number.

First, the total length, $L$, is at least
\[
L \geq \Sigma_j (\Sigma_{i=1}^{m^j} (2f_i^j - 1)) + t^j
   =    \Sigma_j (\Sigma_{i=1}^{m^j} 2f_i^j) - m^j + t^j
.
\]
Since $s$ expresses how far we move to the right and $t$ how
far we move to the left, and the whole path has a bounded number of nodes~$N$, 
we have that $\Sigma_j t^j \geq \Sigma_j s^j - N$.
Thus,
$L \geq (\Sigma_j (\Sigma_{i=1}^{m^j} 2f_i^j) - m^j + s^j) - N$.

$I$ has a number of
complete phases and then some extra requests in addition to that.
There must exist a fixed constant $c$ independent of $I$ such that the cost
of \FIFO on the extra part of any sequence is bounded by $c$.
This follows since there
is a limit of $N$ on how far requests can move to the right.
So if requests never again come so far
to the left that \LRU faults, all requests thereafter
are to only $k$ pages.
This added constant can also take care of the initial extra cost of $k-1$.
Since we are just using a lower bound on the sequence length,
we can ignore the length of a possibly incomplete phase at the end.
Thus,
\[\begin{array}{rcl}
\displaystyle
\frac{\FIFO(I) - \LRU(I)}{|I|}
&
\displaystyle
\leq
&
\displaystyle
\frac{c + \Sigma_j\Sigma_{i=1}^{m^j} (f_i^j - 1) - (s^j - m^j)}{-N + \Sigma_j (\Sigma_{i=1}^{m^j} 2f_i^j) - m^j + s^j} \\[4ex]
&
\displaystyle
\leq
&
\displaystyle
\frac{c + \Sigma_j\Sigma_{i=1}^{m^j} (f_i^j - 1)}{-N + \Sigma_j \Sigma_{i=1}^{m^j} 2f_i^j} \\[4ex]
&
\displaystyle
\leq
&
\displaystyle
\frac{c + \Sigma_j m^j (k - 1)}{-N + \Sigma_j m^j 2k} \\[4ex]
&
\displaystyle
=
&
\displaystyle
\frac{c + (k - 1) \Sigma_j m^j }{-N + 2k \Sigma_j m^j}
\end{array}
\]

The second inequality follows since $s^j\geq m^j$, and
the third inequality follows because
$\frac{f_i^j - 1}{2 f_i^j}\leq\frac{1}{2}$ and $k\geq f_i$ implies
that $\frac{f_i^j - 1}{2 f_i^j}\leq\frac{k-1}{2k}$.

For sequences where the number of phases does not approach infinity,
as argued above, \FIFO's cost will be bounded. For the number of phases
approaching infinity,
$\lim_{j\rightarrow\infty}\frac{c + (k - 1) \Sigma_j m^j }{-N + 2k \Sigma_j m^j}
=\frac{k-1}{2k}=\frac{1}{2}-\frac{1}{2k}$,
which implies the result.

Now, for this proof, we assumed that $s_i<k$.
If $s_i\geq k$, we simply terminate the phase after the processing
of the $s_i$ requests that \LRU faults on,
and continue to define phases inductively from there.
All the bounds from above hold with $t=0$ and the observation
that \FIFO will not fault on the first $s_1$ requests in the next phase.
The direction of the construction is now reversed.
In this process, whenever we reverse the direction as above, we also
rename the variable $s$ to $t$ and $t$ to $s$,
such that $s$ continues to keep track of
movement to the right and $t$ of movement to the left,
and the inequality $\Sigma_j t^j \geq \Sigma_j s^j - N$ still holds.
\end{proof}

\begin{lemma}
For the path access graph $P_N$,
\[\Max^{P_N}(\FIFO,\LRU)=\frac{1}{2}-\frac{1}{2k}.\]
\end{lemma}
\begin{proof}
The upper bound was shown in Lemma~\ref{path-upper-bound}.
Consider the family of sequences
$I_n=\SEQ{1, 2, \ldots, k, k+1, k, k-1, \ldots, 2}^n$.
In each iteration, except the first, \LRU faults twice (on pages $1$
and $k+1$), whereas \FIFO faults on pages $1$ through $k+1$ in every
iteration.
So on this family,
$\lim_{n \rightarrow \infty} \frac{\FIFO(I_n) - \LRU(I_n) }{|I_n|} = 
\frac{k-1}{2k}=\frac{1}{2}-\frac{1}{2k}$,
so the maximum must be at least that large.
\end{proof}

Since \LRU is optimal on paths, this gives :
\begin{theorem}\label{thm:path}
$\I^{P_N}[\FIFO,\LRU ] = [0,\frac{1}{2}-\frac{1}{2k}]$, and \LRU
dominates \FIFO on paths.
\end{theorem}

Note that \FAR and \LRU perform identically on paths, so \FAR
also dominates \FIFO with the same interval.

\section{Star Graphs}
We let $S_N$ denote a star graph with $N$
vertices. A star graph has a central vertex, $s$, which is directly
connected to $N-1$ other vertices, none of which are directly
connected. Thus, we could also see a star graph as a tree with root
$s$ and $N-1$ leaves, all located at a distance one from the root.
We assume that $N\geq k+1$, since otherwise, results become trivial.

\begin{lemma}\label{Star-Min} 
For the star access graph $S_N$,
\[     -\frac{1}{2}+\frac{1}{2(k-1)}
 \leq \Min^{S_N}(\FIFO, \LRU)
 \leq -\frac{1}{2}+\frac{1}{2(k-1)}+\frac{1}{2k(k-1)}
\]
\end{lemma}

\begin{proof}Consider an arbitrary sequence $I$ respecting the
star access graph, and consider its division into $k$-phases.
Since the central vertex occurs after each request to a leaf, each $k$-phase,
except the last, must contain requests to $k-1$ different leaves,
and must be of length at least $2(k-1)$.
As in the proof of Lemma~\ref{any-fifo-lru}, 
\FIFO faults at least once for each of these phases. \LRU faults only
on the leaves and only once on each, so it faults at most
$k-1$ times for each phase.
Thus, if $I$ has $b$ phases, not counting the first empty phase,,
$|I| \geq  2(k-1)(b-1) + 1$
and $\FIFO(I) - \LRU(I) \geq (b-1) - (k-1)(b-1) - k = -(k-2)(b-1) - k$,
and so
$\Min^{S_N}(\FIFO, \LRU ) \geq -\frac{k-2}{2(k-1)} = -\frac{1}{2}+\frac{1}{2(k-1)}$. 

We will show that the upper bound on $\Min^{S_N}(\FIFO, \LRU)$
comes very close to this by analyzing the following sequence.
\begin{align*}
I_n = &~ \SEQ{P, J^n}, ~J= ~ B_1, \ldots , B_{k-1}  \\
 P = &~ \SEQ{1,s, 2, s, \ldots s ,k-2, s, k-1, s, k-2, s, \ldots s , 2, s, 1, s} \\
B_i = &~ \SEQ{k,s, k-1,s, \ldots, s,1,s}, \textrm{ for } 1 \leq i \leq k-1
\end{align*}

We note that $k$ does not appear in $P$ and that all the $B_i$ are identical
(we use the index for reference).
Each $|B_i| = 2k$, so $|I_n| = 2(2k-3) + 2k(k-1)n$.
\LRU starts $B_1$ with a fault on the request to $k$, thereby evicting $k-1$.
It then faults on $k-1$ and evicts $k-2$.
This repeats and ends with the eviction of $k$ at the request to $1$
such that $k-1$ is the least recently used page.
Thus, it faults everywhere except on the central vertex $s$,
which is never evicted by \LRU.
Since \LRU's cache configuration---content as well as the relative ordering
of the recency of pages---is the same at the end of $B_1$ as it was at the
end of $P$, the same pattern must be repeated in each $B_i$.
Thus, $\LRU(I_n) = k + (k-1)kn$. 

\FIFO has three faults in $B_1$: On the request to $k$, where $1$ is evicted,
and at the last two requests of $B_1$. So \FIFO ends $B_1$ with $2$ being
outside its cache. From there onwards, \FIFO faults exactly once in
each $B_i$, $2\leq i\leq k-1$, at the request to $i$,
on which it evicts $i+1$. Therefore,
\FIFO ends each $J$ with $k$ outside its cache
and, hence, the above described fault and eviction pattern is repeated
in every $J$. This gives the cost $\FIFO(I_n) = k + (k+1)n$, and
$\lim_{n \rightarrow \infty} \frac{\FIFO(I_n) - \LRU(I_n)}{|I_n|}$ equals
\[\lim_{n \rightarrow \infty}\frac{k+(k+1)n-(k+(k-1)kn)}{2(2k-3)+2k(k-1)n}
= -\frac{1}{2} + \frac{k+1}{2k(k-1)}\]
Thus,
$\Min^{S_N}(\FIFO, \LRU) \leq -\frac{1}{2} + \frac{k+1}{2k(k-1)}
=-\frac{1}{2}+\frac{1}{2(k-1)}+\frac{1}{2k(k-1)}$.
\end{proof}

\begin{lemma}\label{Star-Max}
For the star access graph $S_N$,
\[ \Max^{S_N}(\FIFO, \LRU) = \frac{1}{4} + \frac{1}{8k-12}. \] 
\end{lemma}

\begin{proof} 
We give a sequence respecting $S_N$ for $N \geq k+1$
giving rise to the stated ratio. Let
\begin{align*}
I_n = \SEQ{P, B^n}, \textrm{ where } P =  \SEQ{1,s,2,s,\ldots,s, k-2,s, k-1,s} \textrm{ and $B$ is}
\end{align*}
\[\left[ \begin{array}{@{\hspace{0.4em}}c@{\hspace{0.4em}}c@{\hspace{0.4em}}c@{\hspace{0.4em}}c@{\hspace{0.4em}}c@{\hspace{0.4em}}c@{\hspace{0.4em}}c@{\hspace{0.4em}}c@{\hspace{0.4em}}c@{\hspace{0.4em}}c@{\hspace{0.4em}}c@{\hspace{0.4em}}c@{\hspace{0.4em}}c@{\hspace{0.4em}}c@{\hspace{0.4em}}c@{\hspace{0.4em}}c@{\hspace{0.4em}}c@{\hspace{0.4em}}c@{\hspace{0.4em}}c@{\hspace{0.4em}}}
 k-2, & s, & \ldots & s, & 2, & s, & 1, & s, & \mathbf{k}, & s, & 1, & s, & 2,  & s, & \ldots & s, & k-2, & s \\
 k-3, & s, & \ldots & s, & 1, & s, & k, & s, & \mathbf{k-1}, & s, & k, & s, & 1, & s, & \ldots& s, & k-3, & s \\
 k-4, & s, & \ldots& s, &k,  & s, &  k-1, & s, & \mathbf{k-2}, & s, & k-1, & s & k, & s, & \ldots&s, & k-4,  & s\\
 \vdots &\vdots  & \ldots & \vdots&\vdots & \vdots & \vdots & \vdots & \vdots & \vdots & \vdots & \vdots &\vdots & \vdots & \ldots & \vdots&  \vdots& \vdots\\
 k, &  s,  & \ldots  & s,  & 4, &s, &3,  & s, &  \mathbf{2}, &s, &3,  & s, &  4,  & s, &\ldots &s, & k, & s \\
 k-1, & s, & \ldots & s, &3,  & s, & 2, & s, & \mathbf{1}, & s, & 2, & s , & 3, & s, &\ldots &s, &  k-1, & s
\end{array} \right] \]
Writing the sequence $B$ like this is just to give an overview.
The sequence is the concatenation of all the rows from top to bottom.

The column in bold indicates the requests that are faults for \LRU.
\LRU faults on exactly one request in every row and so we have
$\LRU(I_n) = k + kn$.
\FIFO faults on $k$ distinct pages in each row,
starting with the request at which \LRU faults.
Thus, $\FIFO(I_n) = k + k^2n$.
Furthermore, $|I_n| = 2(k-1) + (4k-6)kn$.
Since
\[  
  \lim_{n \rightarrow \infty}\frac{\FIFO^{S_N}(I_n) - \LRU^{S_N}(I_n) }{|I_n|}
  =
  \lim_{n \rightarrow \infty}\frac{k + k^2n - (k + kn)}{2(k-1) + (4k-6)kn}
  =
  \frac{k-1}{4k-6},
\]
we have that
$\Max^{S_N}(\FIFO, \LRU) \geq \frac{k-1}{4k-6}=\frac{1}{4} + \frac{1}{8k-12}$.

To prove a tight upper bound on $\Max^{S_N}(\FIFO,\LRU)$,
we consider an arbitrary sequence $I$.
We can assume without loss of generality that $I$ does not contain
any consecutive requests to the same page as they only result in hits
for both algorithms, while increasing the length of the sequence.

We view $I$ as a partition of $k$-blocks with respect to \FIFO, denoted by 
$B_1, \ldots, B_n$, ignoring the first empty block.
Since both \FIFO are \LRU are conservative,
each block, excluding perhaps the last one, 
must have requests to at least $k$ distinct pages.
The access graph is a star, so each request 
must be followed by a request to $s$.
The number of faults incurred by \LRU in $B_i$ is denoted by $\alpha_i$,
where $\alpha_1 = k$. From the maximality of the blocks $B_i$, each
block must have at least one fault for \LRU.
Since $s$ is never evicted from the cache by \LRU,
we have $1 \leq \alpha_i \leq k-1$.

We now find a lower bound on the length of $B_i$.
First recall that \FIFO faults on $k-1$ leaf requests.
We now establish some hits by \FIFO.
Consider a leaf request $r$ that is a fault for \FIFO, but a hit for \LRU.
Since it is not a fault for \LRU, there must have been a request
$r'$ to the same page in the last $k-1$ distinct page requests.
If $r'$ were a fault for \FIFO, then $r$ would have to be a hit.
Since it is not, $r'$ must be a hit for \FIFO.
Since \LRU incurs $\alpha_i$ faults in $B_i$, there are at least $k-1-\alpha_i$
distinct leaf requests where \LRU has a hit while \FIFO faults, ensuring
at least $k-1-\alpha_i$ distinct hits for \FIFO.
Note that even though the hit we establish for \FIFO could be in the
previous block, $B_{i-1}$, it cannot be counted twice, since there are no more
faults on that page after $r'$ in $B_{i-1}$.

The faults and the hits, together with the requests to $s$
following each of them, gives us at least $2(k-1) + 2(k-1-\alpha_i)$
requests. Since the terms not involving $n$ disappear in the limit,
\[\Max^{S_N}(\FIFO, \LRU) 
  \leq
  \max_{\genfrac{}{}{0pt}{}{\alpha_2, \ldots, \alpha_n}{\alpha_i \geq 1}} 
     \left\{ \frac{ \sum_{i=2}^{n} k-\alpha_i}{ 
        \sum_{i=2}^{n-1}(4k - 4 - 2\alpha_i) }
     \right\} 
\]

This is maximized for $\alpha_i=1$ for $2\leq i\leq n$.
Hence, $\Max^{S_N}(\FIFO, \LRU) \leq \frac{k-1}{4k-6}$.
\end{proof}

The algorithms \FAR and \LRU behave identically on star graphs.
Neither of them ever evicts the central vertex. We state the result for both \LRU and \FAR in the
main theorem, though \FAR is not directly mentioned in the lemmas and proofs.

\begin{theorem}\label{thm:S:FIFO}
For the star access graph $S_N$ and $\A \in \{\LRU, \FAR\}$,
\[\begin{array}{rcl}
  \left[ -\frac{1}{2}+\frac{1}{2(k-1)}, \frac{1}{4} + \frac{1}{8k-12} \right]
    & \subseteq & \I^{S_N}[ \FIFO, \A ] \\ & \subseteq &
    \left[ -\frac{1}{2}+\frac{1}{2(k-1)}+\frac{1}{2k(k-1)}, \frac{1}{4} + \frac{1}{8k-12} \right]
\end{array}\]
\end{theorem} 

\begin{proof}
This follows directly from Lemmas~\ref{Star-Min} and~\ref{Star-Max}. \end{proof} 

In \cite{DLM09}, it was shown that
$\Max(\FIFO, \LRU) \geq \frac{k-1}{2k-1}=\frac{1}{2}-\frac{1}{4k-2}$.
The above result shows that for star access graphs, that bound can be
decreased by a factor of approximately two.

Since \LRU and \FAR perform identically on stars, 
$\Min^{S_N}(\FAR, \LRU) =\Max^{S_N}(\FAR, \LRU)= 0$.

The star access graph is another example of where \FWF performs
poorly compared with the other algorithms.

\begin{lemma}\label{Max:Star-FWF}For the star access graph $S_N$, and $\B \in \{\LRU, \FIFO\}$,
\[\Max^{S_N}(\FWF, \B) \leq  \frac{1}{2}.\]
\end{lemma}

\begin{proof}
Given any sequence $I$ in $S_N$, it can be viewed a partition of $k$-phases.
Since it is a star, each phase must be of length at least $2(k-1)$ and
by Lemma~\ref{minimum-cost} \B must incur at least one fault in each phase.
Since \FWF can incur at most $k$ faults in each phase,
if there are $n$ complete phases in $I_n$,
then $\frac{\FWF(I)-\B(I)}{|I|} \leq \frac{n(k-1)}{2n(k-1)} = \frac{1}{2}$. Hence, $\Max^{S_N}(\FWF, \B) \leq  \frac{1}{2}$. 
\end{proof}

\begin{theorem}\label{thm:Max:Star-FWF}
For the star access graph $S_N$, and $\A \in \{\LRU, \FAR, \FIFO\}$,
\[\I^{S_N}[\FWF, \A] = \left[0, \frac{1}{2} \right].\]
\end{theorem}

\begin{proof}
By Lemma~\ref{fwf-conservative},
\[\Min^{S_N}(\FWF, \LRU) = \Min^{S_N}(\FWF, \FIFO) =0.\] 
Furthermore, since \LRU and \FAR perform identically
on star graphs, we also have that $\Min^{S_N}(\FWF, \FAR) = 0$.

Given any sequence $I$ respecting $S_N$, it can be viewed a partition of $k$-phases.
Since $S_N$ is a star, each phase must be of length at least $2(k-1)$, and
\A must incur at least one fault in each phase. Since \FWF can incur at most $k$ faults in each phase, $\lim_{n \rightarrow \infty}\frac{\FWF(I)-\A(I)}{|I|}\leq \frac{k-1}{2(k-1)} = \frac{1}{2}$. Hence, $\Max^{S_N}(\FWF, \A) \leq  \frac{1}{2}$. 

Consider the sequence $I_n = \SEQ{P, (B_1, B_2)^n}$, where $P =  \SEQ{1,s,2,s,\ldots,s, k-2,s, k-1,s}$,
\[
  B_1 = \SEQ{k,s, k-1,s, \ldots,s, 2 , s}, \textrm{ and } 
  B_2 = \SEQ{1, s, 2, s, \ldots,s ,k-1, s}
\]
$B_1$ and $B_2$ have requests to $k$ distinct pages, excluding $1$ and $k$, respectively.

\LRU faults on the first request in each $B_i$. \FWF flushes its cache
at the start of each $B_i$. So $|I_n| = 4(k-1)n + 2(k-1), \LRU(I_n) = 2n+k$
and $\FWF(I_n) = 2kn+k$. So $\lim_{n \rightarrow \infty}\frac{\FWF(I_n)-\LRU(I_n)}{|I_n|} = \frac{2(k-1)}{4(k-1)} = \frac{1}{2}$ and $\Max^{S_N}(\FWF, \LRU) \geq \frac{1}{2}$. 

Let $I_n = \SEQ{P, B^n}$ where  $P =  \SEQ{1,s,2,s,\ldots,s, k-2,s, k-1,s}$ and
\[B =  \left[ \begin{array}{ccccccc} 
k, & s, &k-1 & \cdots& \cdots  & s \\   
1, & s, & k & \cdots  & \cdots  & s \\ 
2, & s, & 1 & \cdots  & \cdots&   s, \\ 
\vdots &\vdots& \vdots  & \cdots & \cdots & \vdots\\ 
\vdots&\vdots  & \cdots & \cdots& \vdots  \\
k-2, &  s & k-3, &\cdots & \cdots  & s,  \\ 
k-1, &  s, & k-2, & \cdots & \cdots & s, \\ 
\end{array} \right] \]

The $i$th row is $i$-free. Hence, each row is of length $2(k-1)$
and $|I_n| = 2(k-1)+ 2(k-1)kn$. Since \FIFO only faults on the first request
in each row, $\FIFO(B) = k$ and
$\FIFO(I_n) = kn+ k$. Since \FWF flushes its cache at the start of each row,
it incurs $k$ faults in each row. Therefore, $\FWF(B) = k^2$ and
$\FWF(I_n) = k^2n+k$. Therefore,
$\lim_{n \rightarrow \infty }\frac{\FWF(I_n)-\FIFO(I_n)}{|I_n|} =
\frac{k(k- 1)}{2(k-1)k} = \frac{1}{2}$ and
$\Max^{S_N}(\FWF, \FIFO) \geq \frac{1}{2}$.
Since \FAR and \LRU behave identically on $S_N$, 
by Lemma \ref{Max:Star-FWF}, we get
$\Max^{S_N}(\FWF, \A) = \frac{1}{2}$.
\end{proof}

\section{Cycle Graphs}
We consider graphs consisting of exactly one cycle,
containing $N$ vertices.
We assume that $N\geq k+1$, since otherwise, results become trivial,
and define $r = N - k$. We 
concentrate on the case where $r < k$, since otherwise the cycle is so large that for the algorithms
considered here, it works as if it were an infinite path.
Thus, for example, there are sequences where \FIFO performs worse than \LRU,
but on worst case sequences, simply going around the cycle,
the algorithms perform identically.
In this section, it is convenient to work modulo $N$ when indexing
pages on the cycle. Thus, if $p < 1$ or $p > N$, we let
$p$ denote the page $p - 1 (\mod N) + 1$.
We will not mention this again later in the proofs to follow.


The following sequences were used in~\cite[Theorem 7]{DLM09} to show that $[ -1 + \frac{1}{k}, \frac{1}{2}- \frac{1}{4k-2}] \subseteq \I[\FIFO, \LRU] $.
\begin{align*}\label{Seq-Im}
I_m = \SEQ{P, B^m }, \textrm{ where } P =  \SEQ{1,2,\ldots, k-1, k}, \textrm{ and $B$ is }
\end{align*}
\[\left[ \begin{array}{cccccccccc}
k-1 & k-2 & \cdots & 2 & 1 & \mathbf{k+1} & 1 & 2 & \cdots & k-1 \\
k-2 & k-3 & \cdots & 1 & k+1 & \mathbf{k} & k+1 & 1 & \cdots & k-2 \\
k-3 & k-4 & \cdots & k+1 & k & \mathbf{k-1} & k & k+1 & \cdots & k-3 \\
\vdots & \vdots & \vdots & \vdots & \vdots & \vdots & \vdots & \vdots & \vdots & \vdots \\
k & k-1 & \cdots & 3 & 2 & \mathbf{1} & 2 & 3 & \cdots & k 
\end{array} \right] \]

\begin{align*}\label{Seq-IM}
I_M = \SEQ{P, B^M}, \textrm{ where } P = \SEQ{1,2,\ldots, k-1, k, k-1, \ldots, 1}, \textrm{ and }
\end{align*}
\[B =  \left[ \begin{array}{cccccc}
{\bf k+1} & k  & k-1 & \cdots & 3 & 2\\
{\bf 1} & k+1  & k & \cdots & 4 & 3\\
{\bf 2} & 1 & k+1 & \cdots & 5 & 4\\
\vdots & \vdots & \vdots & \cdots & \vdots & \vdots\\
{\bf k-1} & k-2 & k-3 & \cdots & k & k+1\\
{\bf k} & k-1 & k-2 & \cdots & 2 &1 
\end{array} \right] \]

These sequences respect $C_{k+1}$, the cycle access graph on $k+1$ vertices. 
Hence, that bound is applicable to cycles of length $k+1$ as well. 

\begin{proposition}\label{prop:implied_results}
For the cycle access graph $C_{k+1}$,  
\[ \I^{C_{k+1}}[\FIFO, \LRU] =[ -1 + \frac{1}{k}, \frac{1}{2}- \frac{1}{4k-2}].\]
\end{proposition}

\begin{proof}
This follows from the results in~\cite{DLM09}, using the sequences
above which respect the cycle, and Lemma~\ref{any-fifo-lru}.
\end{proof}

We now generalize these results to values of $N=k+r$, where $1\leq r\leq k-1$.


\begin{lemma}\label{cycle-FIFO-LRU-min-ub}
For the cycle access graph $C_N$,
\begin{center}
$ \Min^{C_N}(\FIFO, \LRU) \leq -1+\frac{r}{k} $ and $ \Min^{C_N}(\FIFO, \FWF) \leq -1+\frac{r}{k} $ 
\end{center}
\end{lemma}

\begin{proof}
We define
$  J_n = \SEQ{P, B^n }$, where 
$P = \SEQ{1,2, \ldots,k, \ldots N, 1, 2, \ldots, r-1}$
and $B$ is defined by
\[ B = \left[ \begin{array}{rrrr|rrrrr}
 r &  r-1 & \cdots &    1 &  N &  N-1 & \cdots & 2r+2 & 2r+1 \\
2r & 2r-1 & \cdots &  r+1 &  r &  r-1 & \cdots & 3r+2 & 3r+1 \\ 
3r & 3r-1 & \cdots & 2r+1 & 2r & 2r-1 & \cdots & 4r+2 & 4r+1 \\ 
\vdots & \vdots & \cdots & \vdots & \vdots & \vdots & \cdots & \vdots & \vdots \\
 N &  N-1 & \cdots &  k+1 &  k &  k-1 & \cdots & r+2 & r+1 \\
\end{array}\right]\]
The vertical line is merely for reference in the proof.

Let $R$ denote the number of rows in $B$.
(Note that $R=\frac{LCM(N,r)}{r}$, where $LCM(N,r)$ denotes the
least common multiple of $N$ and $r$.)
There are $r$ columns before and $k-r$ columns after the vertical line.
Thus, $|J_n| = N+r-1+kRn$. 

Observe that the sequence turns exactly once,
namely after the first request in $B$.
There are $k-1$ hits following that request for both \FIFO and \LRU.
After that, the sequence moves around the cycle,
so \LRU faults on all of these requests, giving a total cost of
$\LRU^{C_N}(J_n) = N + r - k + kRn$. Note that \FWF faults on the same requests
as \LRU, so $\FWF^{C_N}(J_n) =\LRU^{C_N}(J_n)$.

For \FIFO, when processing $\SEQ{k+1, \ldots, N}$ in $P$, it
evicts $\SET{1, \ldots, r}$,
and then when processing $\SEQ{1, 2, \ldots, r-1}$, it evicts
$\SET{r+1, \ldots, 2r-1}$.
Then, at the very first request of $B$,
it incurs the next fault and evicts $2r$.
After that, the set of pages outside its cache is $\SET{r+1, \ldots, 2r}$,
and \FIFO does not fault again in the first row of $B$.
\FIFO then faults on the first $r$ requests in the second row,
evicting $\SET{2r+1, \ldots, 3r}$.
This pattern continues,
so \FIFO only faults on the first $r$ entries in each row of $B$.
Therefore, $\FIFO^{C_N}(J_n) = N + rRn$. 

This gives
\begin{align*}
  \Min^{C_N}(\FIFO, \LRU)
  \leq &
  \lim_{n \rightarrow \infty}\frac{\FIFO^{C_N}(J_n) - \LRU^{C_N}(J_n)}{|J_n|} \\
  = &
  \lim_{n \rightarrow \infty}\frac{N + rRn - (N + r - k + kRn)}{N+r-1+kRn} \\
  = & - \frac{k - r}{k}
  = -1+\frac{r}{k}.
\end{align*} \end{proof}

\begin{lemma}\label{cycle-FIFO-LRU-max-lb}
For the cycle access graph $C_N$,
\[ \Max^{C_N}(\FIFO, \LRU) \geq \frac{1}{2}-\frac{1}{4k-2} \]
\end{lemma}
\begin{proof}
Let $I_n = \SEQ{S_0,S_1,...,S_n}$, where
\[ S_i = \SEQ{i+k,i+k-1, \ldots ,i+2,i+1,i+2, \ldots, i+k-1,i+k }.\]

Clearly, $\FIFO(S_0)=\LRU(S_0) = k$.

In processing $S_1$, \LRU only faults on $1+k$, where it evicts $1$,
which is not requested in $S_1$.
In general, \LRU faults only on the first request in each $S_i$,
evicting page $i$, which is not requested in $S_i$.
Hence, $\LRU(I_n) = k + n$. 

\FIFO faults on the first request in $S_1$, evicting $k$, which is requested
next. At that request $k-1$ is evicted, leading to a fault on the following
request, etc.
In total, \FIFO faults $k$ times on $S_1$ and pages were brought into
cache in the ordering $i+k$ through $i+1$.
Thus, in general, when the processing of $S_{i+1}$ starts,
the situation repeats.
Hence, we have $\FIFO(I_n) = k + kn$.
The length of the sequence is $|I_n| = (2k-1)(n+1)$.
So,
\begin{align*}
  \Max^{C_N}(\FIFO, \LRU)
  \geq &
  \lim_{n \rightarrow \infty} \frac{\FIFO(I_n) - \LRU(I_n)}{|I_n|} \\
  = &
  \lim_{n \rightarrow \infty} \frac{k + kn - (k + n)}{(2k-1)(n+1)} \\
  = &
  \frac{k-1}{2k-1}
  =
  \frac{1}{2}-\frac{1}{4k-2}
\end{align*}
\end{proof}

\begin{theorem}\label{cycle-FIFO-LRU}
For the cycle access graph $C_N$,
\[
\left[ -1+\frac{r}{k}, \frac{1}{2}-\frac{1}{4k-2} \right]
\subseteq
\I^{C_N}[ \FIFO, \LRU]
\subseteq
\left[ -1+\frac{1}{k}, \frac{1}{2}-\frac{1}{4k-2} \right]
\]
\end{theorem}
\begin{proof}
The left-most containment follows from
Lemmas~\ref{cycle-FIFO-LRU-min-ub} and~\ref{cycle-FIFO-LRU-max-lb},
and the right-most from Lemma~\ref{any-fifo-lru}.
\end{proof}

\begin{theorem}\label{Max:C:FWF}For the cycle access graph $C_{N}$,
\[\I^{C_{N}}[\FWF, \LRU] = \left[ 0, 1 -\frac{1}{k} \right] \]
\end{theorem}

\begin{proof} Sequence, $I_{n} = \SEQ{1,2,\ldots, k, k+1, k , \ldots, 2}^{n}$, respecting $C_{N}$, gives the right endpoint in conjunction with 
Proposition~\ref{general_bound}. The left endpoint is given by Lemma~\ref{fwf-conservative}.
\end{proof}


The exact results to be presented sometimes depend on
the relationship between $k$ and $N$, e.g., whether or
not $r$ divides $N$ (denoted $r \mid N$).
To express many of the results, we need the following term
that, for brevity, we will simply denote $X_r$:
\[ X_r = r(x-1) + \CEIL{\frac{N}{2^x}}, \textrm{ where }
  x = \FLOOR{\log \frac{N}{r}} \]

In the following lemma, we analyze \FAR's behavior on the simplest sequence
exploiting the cycle structure.

\begin{lemma}\label{FAR_WholeCycle}
For \FAR and the sequence $I_{n}= \SEQ{1,2, \ldots, k, \ldots, N}^n$ in $C_N$, 
each $k$-phase, except the first and possibly the last, has $X_r$ faults, and 
\[ \FLOOR{\frac{nN}{k}}X_r +k- X_r \leq \FAR^{C_N }(I_{n})  \leq  \FLOOR{\frac{nN}{k}} X_r + k-1. \] 
\end{lemma}

\begin{proof}In the given sequence, as in any other sequence, the first $k$-phase
contributes $k$ faults. The first phase change in $I_{n}$ occurs at $k+1$,
at which all the other $N-1$ pages are unmarked. Given that the sequence
goes around the cycle $n$ times, without turning, the properties discussed
about faults in the second phase holds for all subsequent ones,
with the possible exception of the last which may contain just one fault.
Consider the fault incurred at the phase change at $k+1$.
The page evicted lies in the middle of the unmarked segment
$[k+2, \ldots, N, 1, \ldots, k]$.
Following this, there are $r-1$ more faults before the next hit.
Each fault leads to the eviction of the page adjacent to the most recently
evicted page, the evictions moving in the same direction
in which the faults are encountered.

In each phase, we refer to the first $r$ faults as the first {\it batch},
faults numbered $r+1$ through $2r$ as the second batch, and so on.
If there are $i$ batches of faults in one $k$-phase, then the first $i-1$
batches will contribute $r$ faults each, and the last batch will have at
least one and at most $r$ faults. For the $i$th batch,
we denote the length of the unmarked segment after marking the first page
in the batch by $d_i$, and the distance to the page evicted at the first
fault in the $i$th batch by $D_i$. These distances are measured in the
direction in which the faulting page was approached. Therefore,
$d_1 = N-1$ and for $i \geq 1$, $d_{i+1}  = d_i - D_i$.
Since \LRU is used to break ties, if for some $i$, $d_i$ is even,
then the closer of the two midpoints is evicted at the first fault of the
$i$th batch. Thus, we have the following dependencies:
\[ \textrm{For $i \geq 1$, } D_i  = \CEIL{d_i/2}
  \textrm{ and }
  d_{i+1} =d_i - D_i= \FLOOR{d_i/2} \]
From the recurrence $d_i = \FLOOR{d_{i-1}/2}$, we obtain the following relation:
\[\textrm{For $i \geq 1$, }
 d_i =  \FLOOR{\frac{d_{i-1}}{2}}
     = \FLOOR{\frac{1}{2}\FLOOR{\frac{d_{i-2}}{2}}}
     = \FLOOR{\frac{d_{i-2}}{2^2}}
     = \FLOOR{\frac{d_1}{2^{i-1}}}
     = \FLOOR{\frac{N-1}{2^{i-1}}}
\]
A $k$-phase ends when all the pages in the cache are marked
and the next request will be a fault.
At any given instant, the marked segment is a path in $C_N$.
This implies that a phase ends when the $r$ pages outside the cache
constitute the unmarked segment, and one of those unmarked pages is
requested. Therefore, if there are $i$ batches
in a $k$-phase, then $d_i + 1 \leq 2r$. Stated differently,
the smallest value of $i$ for which $d_i+1 \leq 2r$
gives the number of batches in a phase. 

If there is an $i$ such that $d_i +1= 2r$, then the phase has $i$ batches
contributing $r$ faults each.  Otherwise, if $d_i + 1< 2r$, then the first
$i-1$ batches contribute $r$ faults each and the last batch contributes
fewer than $r$.

It follows from the above that $d_i + 1 = \CEIL{\frac{N}{2^{i-1}}}$.
Solving $\CEIL{\frac{N}{2^{i-1}}} \leq 2r$ gives
$i-1 = \FLOOR{\log\frac{N}{r}}$ batches with $r$ faults each
and the last with $y=\CEIL{\frac{N}{2^{i-1}}} - r$ faults.
Therefore, each phase in $I_{n}$, excluding the first and perhaps the last,
contains $r(i-1) + y$ faults.
There are $\FLOOR{\frac{nN}{k}}$ complete phases in $I_{r,n}$
and if the last phase is not complete, that is, $k \nmid nN$,
then the last phase can contain at most $r(i-1)+y-1$ faults.
Thus, we obtain the following relation for \FAR serving $I_{n}$:
\[\FLOOR{\frac{nN}{k}} (rx + y) + c \leq \FAR^{C_N}(I_{n})
 \leq
 \FLOOR{\frac{nN}{k}} (rx + y) + rx+y-1 + c,
\]
where $x = \FLOOR{\log\frac{N}{r}}$,
     $y = \CEIL{\frac{N}{2^x}}- r$ and
     $c =  k - (rx + y)$.
\end{proof}

The following lemma analyzes \FAR's behavior on a cycle when the cycle
structure is {\em not} used. Thus, the cycle access graph is used as a path access graph. However,
\FAR is oblivious to this and uses distances involving the non-utilized
edge in the graph, leading to non-optimal results.

From now on, whenever needed , we use $\hat{N}$ to denote $N$, if $N$ is even, and $N-1$, otherwise.

\begin{lemma}\label{LRU_FAR_Min_Ubound_Cycle2}
For \FAR and the sequence $I_{n}= \SEQ{1,2, \ldots, k, \ldots, N-1, N, N-1, \ldots, 2}^n$ in $C_N$, each $k$-phase, except the first (which has $k$) and the last (which has $r$), has $rx+y$ faults, where $x = \FLOOR{\log\frac{\hat{N}}{r}}$ and $y = \FLOOR{\frac{\hat{N}}{2^x}}-r$.
\end{lemma}
\begin{proof}
The first $k$-phase in $I_{n}$ has $k$ faults. In any $k$-phase of $I_{n}$, excluding the first, the first set of $r$ faults is called the
first {\it batch}, faults numbered $r+1$ through $2r$ is called
the second batch, and so on. If there are $i$ batches of faults
in one $k$-phase, then the first $i-1$ batches will contribute $r$ faults
each, and the last batch will have at least one and at most $r$ faults.

As before, the length of the unmarked segment after marking the
first page of the $i$th batch is denoted by $d_i$ and the page located
$D_i$ pages away is evicted at that fault. All these distances are measured
in the direction in which the first fault of the batch was encountered.
Note that within each iteration within $I_n$, there are two phase changes,
occurring first at $k+1$ and then at $r$. In the following discussion,
we explain the behavior of \FAR in one iteration within $I_n$.
Since the same properties hold for others, that will lead to a bound
for $\FAR^{C_{N}}(I_{n})$. 
 
At the end of a phase and right before the start of the next, {\FAR}'s cache is connected. Hence, the $r$ pages outside the cache also form a connected component, implying that the sets of pages outside \FAR's cache immediately before the phase changes at $k+1$ and $r$ are $\SET{k+1, \ldots, N}$ and $\SET{r, r-1, \ldots, 1}$, respectively.
 
For the phase changes at $k+1$ and $r$, the faulting request is approached
from $k$ and $r+1$, respectively. For either case, we have $d_1 = N-1$
and as in Lemma~\ref{FAR_WholeCycle}, the page located $D_1 = \CEIL{d_1/2}$
vertices away is evicted at the first fault in the phase.
The next $r-1$ faults lead to eviction of pages in the same direction
in which the faults are encountered. Unlike in the previous lemma,
the sequence considered here turns back at the end of the first batch
and so the second batch of faults start at the most recently evicted page.

{\it Phase change at $k+1$}:
The first fault in the second batch occurs when the sequence reaches
$D_1 $, which is also the first page marked in the batch.
The unmarked segment at that instant is $\{ D_1,D_1 -1, \ldots, 1\}$.

{\it Phase change at $r$}: Analogously to the previous case,
the second batch of faults starts when the sequence reaches $N-D_1+1$.
The unmarked segment at that instant is
\[ \Big[N-D_1+2 , N-D_1 +3, \ldots, N-D_1 +r, \ldots, k, k+1, \ldots, N-1 , N \Big].\]

In either case, the length of the unmarked segment is $d_2 =D_1 -1$.
Note that for both locations of phase change, the change in direction of
the sequence right after the first batch affects the resolution of ties
in subsequent batches. In fact, if $d_2$ is even, then the farther of the
two midpoints, measured in the same direction as the fault, is less recently
requested than the other. Therefore, for each phase, we have the following
correspondence:
\[ D_2 = \begin{cases} d_2/2 +1,     & \hbox{if $d_2$ is even } \\
                      \CEIL{d_2/2}, & \hbox{if $d_2$ is odd }
        \end{cases}
\]
Since, in either case, from the second batch onwards, the sequence does not
change direction for the rest of the phase, all subsequent ties within the
phase are resolved in the manner of the second batch. 
Therefore, in any given phase, from the second batch onwards, if the
unmarked segment is even, the farther of the two midpoints, measured in
the same direction in which the fault was approached is evicted in favor
of the other. This yields the following set of relations:
$d_1 = N-1$, $D_1 = \CEIL{d_1/2}$, $d_2 = D_1 -1$, and for $i \geq 2$,
\begin{align*}
D_i = & \begin{cases}
         d_i/2 +1,     & \hbox{if $d_i$ is even }\\
         \CEIL{d_i/2}, & \hbox{if $d_i$ is odd }\\
       \end{cases}
\end{align*}
and
\begin{align*}
d_{i+1} =  d_i -D_i = &
       \begin{cases}
         d_i/2 -1,      & \hbox{if $d_i$ is even }\\
         \FLOOR{d_i/2}, & \hbox{if $d_i$ is odd }\\
       \end{cases} 
\end{align*}

This implies that for all $i \geq 2$, $d_{i+1}=~\FLOOR{\frac{d_i-1}{2}}$. 

We now establish the following claim.
Recall that $\hat{N}$ denotes $N$, if $N$ is even, and $N-1$, otherwise. 

\begin{claim}For $i \geq 3$, we have $d_{i} + 1 = \FLOOR{\frac{D_1}{2^{i-2}}} = \FLOOR{\frac{\hat{N}}{2^{i-1}}}$. 
\end{claim}

\begin{proof}
Since $D_{1} = \CEIL{\frac{N-1}{2}}$, using the new notation,
$D_{1} = \frac{\hat{N}}{2}$.
 
We proceed to show by induction that for $i\geq 3$,
$d_{i} + 1 = \FLOOR{\frac{D_1}{2^{i-2}}}$.

For the base case, $i=3$, we have
\[
 d_{3} =\FLOOR{\frac{d_{2}-1}{2}} = \FLOOR{\frac{(D_{1} -1) -1 }{2}} =  
\FLOOR{\frac{D_{1}}{2}}-1.
\]
Hence,  $d_{3} +1 = \FLOOR{\frac{D_1}{2^{3-2}}}$.
 
Now, we assume that the induction hypothesis holds up to some $i \geq 3$.
For the induction step, we prove the relation
$d_{t+1} =\FLOOR{\frac{d_{t}-1}{2}}$, by applying the hypothesis for $d_{t}$
in the last equality below.
\begin{align*}
d_{t+1}= &~\FLOOR{\frac{d_{t}-1}{2}}= \FLOOR{\frac{1}{2}(d_{t}+1) -1}=\FLOOR{\frac{1}{2}\FLOOR{\frac{D_1}{2^{t-2}}} }-1
\end{align*}
Therefore, $ d_{t+1} + 1= \FLOOR{\frac{D_1}{2^{t-1}}}$,
and the claim is proved.
\end{proof}

As was the case in the previous lemma,
the last batch starts when for the first time in the current phase,
the length of the unmarked segment is no greater than $2r$, i.e.,
the smallest value $i$ for which $d_i+1 \leq 2r$
gives the number of batches in the phase.
Solving $\FLOOR{\frac{\hat{N}}{2^{i-1}}} \leq 2r$
gives $i-1 = \FLOOR{\log\frac{\hat{N}}{r}}$.
Therefore, the first $i-1$ batches in a $k$-phase have $r$ faults each.
In the last batch, though, there 
are exactly $\FLOOR{\frac{\hat{N}}{2^{i-1}}} -r$ faults.

Right before the start of the $i$th batch, the length of the unmarked
segment is $d_i+1$. The phase must end when the length of the unmarked
segment becomes $r$. Therefore, $d_i+1-r$ is an upper bound on the number
of faults incurred in the $i$th batch.
\end{proof}

Note that in the above proof, making the sequence go only up
to some other value between $k+1$ and $N-1$, instead of up to $N$,
would never give more faults.

\begin{lemma}\label{MaxDist_FaultToHole_InABatch}
For $1\leq r \leq k-1$, in any sequence respecting the cycle access graph $C_N$,
the maximum number of faults incurred by \FAR in a $k$-phase, excluding the
first, is at most $X_r$. In particular, \FAR incurs the maximum number of faults in a $k$-phase
if the sequence takes the shortest path between any two faults in that phase. 
Consequently, in $C_{k+1}$, each $k$-phase can generate at most
$\CEIL{\log (k+1)}$ faults for \FAR.
\end{lemma}
\begin{proof}
Given the eviction rule of \FAR in $C_N$, which is that it evicts
the midpoint of the current unmarked segment, it follows that when a
sequence does not turn inside a phase, it is taking the shortest path
to the next fault. This situation is analyzed in Lemma~\ref{FAR_WholeCycle}.
When a sequence turns such that at least one page is marked before the next
turn, then all those pages become unavailable for eviction for the
remainder of the phase. A phase ends when all the pages in the cache
are marked and a new phase starts at the next fault. Therefore, if a
sequence keeps moving along the shortest path which takes it to the
next fault, then it is also marking the fewest number of pages in order
to get to the next fault, thereby, maximizing the number of faults
\FAR incurs in the current phase. Hence, the maximum number of faults
incurred by \FAR in each phase, excluding the first, is upper bounded
by $X_r$, as proved in Lemma~\ref{FAR_WholeCycle}. The special case of 
$C_{k+1}$ is given by $r=1$ and so the lemma is proved.
\end{proof}

\begin{lemma}\label{LB:C:Min}For the cycle access graph $C_N$, and $\A \in \{\LRU, \FIFO, \FWF\}$,
\[ \Min^{C_N}(\A, \FAR) \geq -\frac{X_r -1}{k}.\]
\end{lemma}

\begin{proof}
Consider an arbitrary sequence $I_n$ in $C_N$, where $n$ denotes the number
of $k$-phases in the sequence. The last phase of a sequence may contain fewer 
than $k$ distinct pages and in that case we can ignore the last phase in $I_n$.
Note that each phase contains requests to $k$ distinct pages. It follows
that each phase in a sequence is of length at least $k$ and \A incurs at
least one fault in each of them. By Lemma~\ref{MaxDist_FaultToHole_InABatch},
we know that \FAR can incur at most $X_r$ faults in each phase,
excluding the first. By Lemma~\ref{minimum-cost}, \A faults at least once in each phase.
In the first phase both algorithms incur $k$ faults.
Thus, in each phase,
the absolute value of the maximum
difference in faults is at most $X_r-1$. 
Thus,
$\lim_{n \rightarrow \infty}\frac{\A(I_n) - \FAR(I_n)}{|I_n|}
\geq -\frac{X_r-1}{k}$. \end{proof}

\begin{lemma}\label{large-cycle1}
For the cycle access graph $C_N$,
\[ \Min^{C_N}(\FIFO, \FAR) \leq -\frac{ X_r  -r }{k}. \] 
\end{lemma}


\begin{proof} Recall the sequence $J_n$ from the proof of Lemma~\ref{cycle-FIFO-LRU-min-ub}.
\[J_n = \SEQ{P, B^n}, \textrm{ where }
   P = \SEQ{1,2, \ldots, k, k+1, \ldots, N, 1, 2, \ldots, r-1}
\]
and
\[ B=\left[ \begin{array}{llll|lllll}
r  & r-1 &  \cdots  & 1 & N & N-1 &  \cdots  & 2r+2  & 2r+1   \\
2r & 2r-1 &  \cdots  & r+1 & r & r-1 &  \cdots  & 3r+2 & 3r+1 \\ 
3r & 3r-1 &  \cdots  & 2r+1 & 2r & 2r-1 &  \cdots  & 4r+2 & 4r+1 \\ 
 \vdots  &   \vdots  &   \cdots  &    \vdots  &    \vdots  &   \vdots  &  \cdots  &    \vdots  &  \vdots  \\
N  &  N-1  &   \cdots  &   k+1 &  k  &   \vdots  &  \cdots  &   r+2 & r+1  \\
\end{array}\right]\]

$|J_n| = kRn + N+r-1$ and $\FIFO^{C_N}(J_n) = N+rRn$. 

There is exactly one turn in $J_n$,
which occurs at the first request in $B$ and nowhere else.
For the rest of the sequence, it moves around the cycle without turning.
Hence, the number of faults incurred by \FAR in each phase of $B$,
excluding the first two, is given by Lemma~\ref{FAR_WholeCycle},
to be $X_r = r\big(x -1\big) + \CEIL{\frac{N}{2^x}}$,
where $x = \FLOOR{\log \frac{N}{r} }$.
Therefore, $\FAR^{C_N}(J_n) = \FLOOR{\frac{nkR}{k}}X_r + c$,
where $c$ is a constant.
The constant bounds the number of faults in the first two phases. 
Now, $\Min^{C_N}(\FIFO,\FAR)$ is at most
\[\lim_{n \rightarrow \infty}\frac{\FIFO^{C_N}(J_n) - \FAR^{C_N}(J_n)}{|J_n|} = \lim_{n \rightarrow \infty}\frac{nR(r - X_r)}{nRk +~N+r-1} = -\frac{X_r-r}{k} \]
\mbox{}\end{proof}

\begin{lemma}\label{Max:C:A:FAR}For the cycle access graph $C_N$, and $\A \in \{\LRU, \FIFO, \FWF\}$, 
\[ \Max^{C_N}(\A, \FAR) \geq 1-\frac{X_r}{k}. \]
\end{lemma}
\begin{proof}
Consider the sequences $I_{n} = \SEQ{ 1, 2, \ldots, N}^n$ in $C_N$
such that $k$ divides $nN$. 
It is easy to see that $\A(I_{n}) = |I_{n}|= Nn$. By
Lemma~\ref{FAR_WholeCycle}, we have
$\FAR^{C_N}(I_{n}) \leq \frac{Nn}{k} X_r + k-1$.
Thus, $\Max^{C_N}(\A, \FAR)$ is at least
\[ \lim_{n \rightarrow \infty}\frac{ \A^{C_N}(I_{n}) - \FAR^{C_N}(I_{n})}{|I_{n}|} \geq \lim_{n \rightarrow \infty}\frac{nN - \frac{nN}{k}X_r-(k-1)}{nN}  = 1-\frac{X_r}{k} \]
\end{proof}

\begin{lemma}\label{Min:C:LRU:FAR}For the cycle access graph $C_N$,

\[ \Min^{C_N}(\LRU, \FAR) \leq -\frac{r\left( \FLOOR{ \log\frac{\hat{N}}{r}}  -1 \right) }{N-1} \]

where $\hat{N}$ is $N$ and $N-1$ if $N$ is even and odd, respectively.
\end{lemma}

\begin{proof}
Consider the sequence $I_{n} = \SEQ{1, 2,\ldots, N-1, N, N-1, \ldots, 2}^{n}$ used in the proof of
Lemma~\ref{LRU_FAR_Min_Ubound_Cycle2}. Clearly, 
$\LRU^{C_{N}}(I_{n}) = 2nr+k-1$ and $|I_{n}| = 2(N-1)n $.
There are two phase changes in each iteration of $I_n$, so
by Lemma~\ref{LRU_FAR_Min_Ubound_Cycle2},
\[ 
k + 2n \left(r\FLOOR{ \log\frac{\hat{N}}{r}} +  \FLOOR{\frac{\hat{N}}{2^x}} -r \right)  \leq \FAR^{C_N}(I_{n}), \]
where $x = \FLOOR{ \log\frac{\hat{N}}{r}}$.

Now, since $r \leq \FLOOR{\frac{\hat{N}}{2^x}}$,
\[\lim_{n \rightarrow \infty}\frac{ \LRU^{C_N}(I_{n}) - \FAR^{C_N}(I_{n}) }{|I_{n}|} \leq -\frac{r \left(
\FLOOR{ \log\frac{\hat{N}}{r}}  -1 \right)   }{N-1 }
\]
\end{proof}

When $r=1$, we
get the bound
$\Min^{C_{k+1}}(\LRU, \FAR) \leq  - \frac{ \FLOOR{\log k} -1 }{k}$.

\begin{theorem}\label{thm:C:FAR}
For the cycle access graph $C_N$,
\[
\left[-\frac{X_r  -r}{k}
 , 1-\frac{X_r}{k} \right] \subseteq \I^{C_N}[ \FIFO, \FAR]  \subseteq \left[- \frac{X_r-1}{k} , 1-\frac{1}{k} \right],
\]
\[
\left[- \frac{r\left(  \FLOOR{ \log \frac{\hat{N}}{r} }   -1 \right) }{N-1},   1-\frac{X_r}{k}  \right ] \subseteq \I^{C_N}[ \LRU, \FAR] \subseteq \left[- \frac{X_r-1}{k} , 1-\frac{1}{k} \right]  \\
\]
and
\[
\left[0, 1-\frac{X_r}{k} \right] \subseteq \I^{C_N}[ \FWF, \FAR] \subseteq  \left[0, 1 -\frac{1}{k} \right].
\]
\end{theorem}
\begin{proof}
The first relation follows from Proposition~\ref{general_bound} and Lemmas~\ref{LB:C:Min}, \ref{large-cycle1}, and \ref{Max:C:A:FAR},
and the second from Proposition~\ref{general_bound} and
Lemmas~\ref{LB:C:Min}, \ref{Min:C:LRU:FAR}, and \ref{Max:C:A:FAR}.
The third result follows from Proposition~\ref{general_bound} and Lemmas~\ref{fwf-conservative} and \ref{Max:C:A:FAR}. 
\end{proof}

\section{Concluding Remarks}
Relative interval analysis has the advantage that it can
separate algorithms properly when one algorithm is at least
as good as another on every sequence and is better on some.
This was reflected in the results concerning \FWF which is
dominated by the other algorithms considered for all access
graphs. It was also reflected by the result showing that
\LRU and \FAR have better performance than \FIFO on paths.
The analysis also found the expected result that \FAR, which is
designed
to perform well on access graphs, performs better than both \LRU
and \FIFO on cycles.

However, it is disappointing that the relative interval
analysis of \LRU and \FIFO on stars and cycles found that
\FIFO had the better performance, confirming the original results
by~\cite{DLM09} on complete graphs. Clearly, the access graph
technique cannot be arbitrarily applied to all quality measures
for online algorithms to show that \LRU is better than \FIFO.
To try to understand quality measures better,
it would be interesting to determine on which
the access graph technique
is useful for this well studied problem and on which it is not.

\bibliographystyle{plain}
\bibliography{refs}

\end{document}